\documentclass[a4paper,12pt]{article}
\usepackage{amsmath,amssymb,amsthm}
\usepackage[english]{babel}
\usepackage{hyperref}
\title{The model equation of soliton theory}
\author{V.E. Adler, A.B. Shabat}
\date{\small Landau Institute for Theoretical Physics, RAS,\\
 1A pr. Ak. Semenova, 142432 Chernogolovka, Russia. \\
 E-mail: {\tt adler@itp.ac.ru, shabat@itp.ac.ru}}
\advance\textwidth 20mm  \advance\oddsidemargin -10mm
\advance\textheight 20mm \advance\topmargin -15mm

\def\Integer{\mathbb Z}

\def\pd  {\partial}
\def\ti  {\widetilde}

\def\a{\alpha}
\def\b{\beta}
\def\g{\gamma}
\def\d{\delta}
\def\eps{\varepsilon}
\def\ka{\kappa}
\def\la{\lambda}
\def\ph{\varphi}

\def\const {\mathop{\rm const}\nolimits}
\def\tr    {\mathop{\rm tr}\nolimits}

\def\diag  {\mathop{\rm diag} \nolimits}

\def\<{\langle}
\def\>{\rangle}
\def\defeq {\stackrel{\mbox{\rm\small def}}{=}}

\numberwithin{equation}{section}

\theoremstyle{plain}
  \newtheorem{theorem}{Theorem}
  \newtheorem{lemma}[theorem]{Lemma}
  \newtheorem{statement}[theorem]{Statement}
  \newtheorem{corollary}[theorem]{Corollary}
\theoremstyle{definition}
  
\theoremstyle{remark}
  \newtheorem{example}{Example}

\begin{document}\maketitle
\thispagestyle{empty}

\begin{abstract}
We consider an hierarchy of integrable $1+2$-dimensional equations related
to Lie algebra of the vector fields on the line. The solutions in
quadratures are constructed depending on $n$ arbitrary functions of one
argument. The most interesting result is the simple equation for the
generating function of the hierarchy which defines the dynamics for the
negative times and also has applications to the second order spectral
problems. A rather general theory of integrable $1+1$-dimensional equations
can be developed by study of polynomial solutions of this equation under
condition of regularity of the corresponding potentials.
\end{abstract}

%-------------------------------------------------------------------------------
\section{Introduction}

In this paper we study the equation
\begin{equation}\label{sme}
 D_\tau G(\la)=\frac{\<G(\la),G(\mu)\>}{\la-\mu}
\end{equation}
where the unknown function $G(\la)=G(x,\tau,\la)$ is locally analytic on
$\la$ and
\[
 \<f,g\>\defeq fg_x-gf_x.
\]
The assumption on the analyticity is essential since otherwise the general
solution contains too much arbitrariness (cf Example \ref{ex:x^2} below).

The solution of the Cauchy problem with initial data $G_0(\la)=G_0(x,\la)$
can be found as the Taylor expansion
\begin{equation}\label{taun}
 G(\la)=G_0(x,\la)+\tau G_1(x,\la)+\tau^2G_2(x,\la)+\dots
\end{equation}
The substitution into equation yields the recurrence relation
\[
 (n+1)G_{n+1}(\la)=\frac{1}{\la-\mu}\sum_{k=0}^n\<G_k(\la),G_{n-k}(\mu)\>,\quad
 n=0,1,\dots
\]
which allows to compute all coefficients if $G_0$ is a smooth function on
$x$ and $\la$. Notice that if $G_0$ is a polynomial in $\la$ of degree $m$
then this recurrence relation implies that all $G_k(\la)$ are polynomials of
degree not greater $m$ as well. It is not difficult to obtain an estimation
for all coefficients of the expansion and to prove that the convergence
radius is not zero. Thus, the dynamics on $\tau$ preserves the polynomial
structure which correspond to a class of rather interesting  solutions.

The differentiation $D_\tau$ obviously depends on the choice of $\mu$ in
equation (\ref{sme}). It turns out that this arbitrariness leads to the
mutually consistent dynamical equations, as the following theorem states.

\begin{theorem}\label{th:D12}
Let differentiation $D_i$ is defined by the equation (\ref{sme}) at
$\mu=\mu_i$. Then $[D_1,D_2]=0$.
\end{theorem}
\begin{proof}
Denote $g_i=G(\mu_i)$, then equation (\ref{sme}) implies
\[
 D_1g_2=\frac{\<g_2,g_1\>}{\mu_2-\mu_1}=\frac{\<g_1,g_2\>}{\mu_1-\mu_2}=D_2g_1.
\]
It is easy to check that the Jacobi identity
\[
 \<\<G,g_1\>,g_2\>+\<\<g_2,G\>,g_1\>+\<\<g_1,g_2\>,G\>=0
\]
implies the equality $[D_1,D_2](G)=0$.
\end{proof}

It is clear that the Theorem \ref{th:D12} remains valid if we replace
$\<,\>$ by the bracket in an arbitrary Lie algebra. Our original bracket
corresponds to the Lie algebra of the vector fields on the line, another
interesting example is the bracket $\<f,g\>=f_yg_x-f_xg_y$, corresponding to
the Hamiltonian vector fields on the plane. However, we restrict ourselves
by the simplest case since it is already quite nontrivial and illustrative.

The above equation for the functions $g_1$, $g_2$ can be written as an
equation with partial dirivatives $\pd_{\tau_i}=D_i$ for the potential $u$
which is introduced by equations $g_i=u_{\tau_i}$:
\begin{equation}\label{mu12}
 (\mu_i-\mu_j)u_{\tau_i\tau_j}=u_{\tau_i}u_{x\tau_j}-u_{\tau_j}u_{x\tau_i}.
\end{equation}
This differential equation (cf \cite{alonso_shabat}) can be considered as a
simplified version of the original problem, since the analyticity of the
solution on the parameter here is not important. The particular solutions of
this equation are presented in the Section \ref{s:trace}. Obviously,
introducing a new independent variable $\tau_k$ corresponding to $\mu_k$
leads to equation consistent with (\ref{mu12}). Some generalizations with
similar symmetry properties are discussed in the Section \ref{app:triples}.
It would be interesting to compare these examples with the classification
results obtained in the papers \cite{FKPT} in the framework of the method of
hydrodynamic reductions.

Thus, the important feature of the equation (\ref{sme}) is its symmetry
properties. In the case of polynomial in $\la$ solutions the sufficient set
of symmetries allows to integrate equation (\ref{sme}) completely. The
general solution contains, in the case of $n$-th degree polynomials, $n$
arbitrary functions of one variable $\ka_i(\la)$ and the problem is reduced
to the interpolation of the polynomial in the given set of points
$G_x|_{\la=\g_i}=\ka_i(\g_i)$ (see Section \ref{s;polsol}). The zeroes
$\g_i$ of the polynomial $G(\la)$ play the role of the Riemannian invariants
for the system of hyperbolic equations on the coefficients of $G(\la)$ (cf
\cite{EF}).

In the Section \ref{s:flows} we show that equation (\ref{sme}) rewritten in
terms of Laurent expansions of the function $G(\la)$ generate an infinite
sequence of commuting vector fields which can be interpreted as the
additional symmetries of equations (\ref{mu12}).

The Appendix is devoted to the applications of equation (\ref{sme}) to the
spectral problems of the second order
\[
 \psi_{xx}=U(x,\la)\psi.
\]
We discuss there the problem of construction of $G(x,\la)$ for a wide class
of potentials $U(x,\la)$. The one-to-one correspondence between the function
$G(x,\la)$ and the potential $U(x,\la)$ can be achieved if we waive the
polynomiality in $\la$. In the direct problem this is equivalent to solving
of Riccati equation, while the solution of the inverse problem defines the
potential $U(x,\la)$ as the Schwarz derivative of $1/G(x,\la)$ with respect
to $x$.

The equation (\ref{sme}) admits not only polynomial in $\la$ solutions, but
also the rational ones. The generalization of Dubrovin equations for this
case and analysis of possible applications to the KdV-like equations is an
interesting open problem.

We finish this introduction by an example which demonstrates the importance
of the choice of the suitable analytical structure of $G(\la)$.

\begin{example}\label{ex:x^2}
The equation (\ref{sme}) admits the reduction
\begin{equation}\label{abc}
 G=a(\tau,\la)x^2+2b(\tau,\la)x+c(\tau,\la)
\end{equation}
which leads to the system
\[
 \frac{da}{d\tau}=2\{a,b\},\quad \frac{dc}{d\tau}=2\{b,c\},\quad
 \frac{db}{d\tau}=\{a,c\},\quad
 \{f,g\}\defeq \frac{f(\mu)g(\la)-g(\mu)f(\la)}{\la-\mu}.
\]
It is easy to see that the discriminant $b^2-ac$ does not depend on $\tau$.
The equation (\ref{mu12}) turns into the matrix equation
\[
 (\mu_i-\mu_j)S_{\tau_i\tau_j}=[S_{\tau_i},S_{\tau_j}],\quad
 S=\begin{pmatrix} ~b & ~a \\ -c & -b \end{pmatrix}\in sl_2
\]
which possesses the partial integrals
$\pd_{\tau_j}(S_{\tau_i},S_{\tau_i})=0$ where $(S_1,S_2)=\tr S_1S_2$.

In particular, at $a=0$ the equation becomes linear:
\[
 \frac{df(\la)}{dt}=\frac{f(\la)-f(\mu)}{\la-\mu},\quad
 \frac{d}{d\tau}=2b(\mu)\frac{d}{dt},\quad
 G(\tau,\la)=2b(\la)(x+f(\tau,\la)).
\]
\end{example}

%-------------------------------------------------------------------------------
\section{The commuting vector fields}\label{s:flows}

We will assume that the locally analytic functions under consideration are
regular at $\la=0$ (this can be achieved by a shift) and that the infinity
is not the essential singular point. Then the expansions exist
\begin{equation}\label{ab}
 G(\la)=\begin{cases}
  A(\la)=a_0\la^m+a_1\la^{m-1}+a_2\la^{m-2}+\dots,\quad \la\to\infty,\\
  B(\la)=b_0+b_1\la+b_2\la^2+\dots,\quad \la\to0
 \end{cases}
\end{equation}
which converge at $|\la|>\rho_1$ and $|\la|<\rho_2$ respectively.

Let us assume that $\mu$ belongs to the domain of convergence of one of the
power series (\ref{ab}) and rewrite the equation (\ref{sme}) as follows
\[
 (\la-\mu)D_\tau A(\la)=\<A(\la),G(\mu)\>.
\]
Collecting the coefficients at $\la^{m+1}$ one obtains $D_\tau a_0=0$. The
change of the form $\pd_x\to\phi(x)\pd_x$, $G\to G/\phi(x)$ allows to set
$a_0=1$ without loss of generality. Then
\begin{equation}\label{aag}
 D_\tau a_n=\<A_n(\mu),\, G(\mu)\>,\quad
 A_n(\la)=\la^n+a_1\la^{n-1}+a_2\la^{n-2}+\dots+a_n.
\end{equation}
Analogously, the equation
\[
 (\la-\mu)D_\tau B(\la)=\<B(\la),G(\mu)\>
\]
in the neighborhood of $\la=0$ implies
\[
 D_\tau b_n=\<A_{-n-1}(\mu),G(\mu)\>,\quad
 A_{-1}(\mu)=-\frac{b_0}{\mu},\quad
 A_{-2}(\mu)=-\frac{b_0}{\mu^2}-\frac{b_1}{\mu},~\dots
\]
Therefore, the expansions (\ref{ab}) and the equation (\ref{sme}) bring to
an infinite sequence of the polynomials $A_n$:
\begin{equation}\label{aba}
 A_n(\la)=\la A_{n-1}(\la)
  +\left\{\begin{array}{ll}
       a_n, & n>0\\
    b_{-n}, & n<0
   \end{array},\right.\quad
 A_0=1,\quad A_{-1}(\la)=-\frac{b_0}{\la}.
\end{equation}
The polynomials $A_n$, $n\in\Integer$ and Theorem \ref{th:D12} allow to
rewrite the original equation (\ref{sme}) in the form of a system of
equations for the coefficients of the series $A(\la)$ and $B(\la)$. In order
to do this we introduce, following the papers \cite{alonso_shabat}, the
sequence of the vector fields ${\cal L}_n$, $n\in\Integer$ as the
differential operators
\begin{equation}\label{dan}
 {\cal L}_n= D_n-A_n(\la)D_0,\quad D_n=\frac{\pd}{\pd t_n},\quad
 D_0=\frac{\pd}{\pd x}.
\end{equation}
The hierarchy of the corresponding times will be denoted as ${\bf
t}=(\dots,t_{-1},t_0,t_1,\dots)$.

One can prove, by use of the above equations for $D_\tau a_n$ and $D_\tau
b_n$, that the commutativity of the differentiations from the Theorem
\ref{th:D12} leads to the commutation relations
\begin{equation}\label{da}
 [{\cal L}_m,{\cal L}_n]=0 \quad\Leftrightarrow\quad D_mA_n-D_nA_m=\<A_m,A_n\>,
 \quad \forall~m,n\in\Integer.
\end{equation}
Finally, these relations between the polynomials (\ref{aba}) can be
rewritten intermediately, as equations for $G(\la)$:
\begin{equation}\label{dna}
 D_nG(\la)=\<A_n(\la),\,G(\la)\>,\quad G(\la)=G({\bf t},\la),\quad n\in\Integer.
\end{equation}
The equivalence of (\ref{da}) and (\ref{dna}) is proved along the standard
scheme \cite{ush}.

Substitution of the series (\ref{ab}) into equation (\ref{dna}) yields
\begin{equation}\label{dnab}
 D_n A(\la)=\<A_n(\la),A(\la)\>,\quad
 D_n B(\la)=\<A_n(\la),B(\la)\>,\quad n\in\Integer.
\end{equation}
In particular, for the first ``negative'' time $t_{-1}$ we have
\[
 D_{-1}A(\la)=\frac{b_{0,x}}{\la}+\frac{\<a_1,b_0\>}{\la^2}
  +\frac{\<a_2,b_0\>}{\la^3}+\dots,\quad
 D_{-1}B(\la)=\<b_1,b_0\>+\<b_2,b_0\>\la+\dots.
\]
This leads to an infinite autonomous system of equations for the variables
$b_i$:
\[
 D_{-1}b_0=\<b_1,b_0\>,\quad D_{-1}b_1=\<b_2,b_0\>,~\dots
\]
which correspond to the value of $\mu$ chosen in the convergence domain of
the series $B(\la)$.

Analogously, for the first ``positive'' flow $D_1=\pd_{t_1}$, the basic
equations (\ref{dna}) give
\[
 D_1A(\la)=\<\la+a_1,\frac{a_2}{\la^2}+\frac{a_3}{\la^3}+\dots\>,\quad
 D_1B(\la)=\<a_1,b_0\>+ (b_{0,x}+\<a_1,b_1\>)\la+\dots,
\]
and this is equivalent to an infinite sequence of equations for the
variables $a_i$:
\begin{equation}\label{dta}
 D_1a_1=a_{2,x},\quad
 D_1a_2=a_{3,x}+\<a_1,a_2\>,\quad D_1a_3=a_{4,x}+\<a_1,a_3\>,\dots
\end{equation}
corresponding to the value of $\mu$ in the convergence domain of the series
$A(\la)$.

%-------------------------------------------------------------------------------
\section{Polynomial solutions}\label{s;polsol}

In the polynomial case $G=\la^n+g_1\la^{n-1}+\dots+g_n$ the infinite
hierarchy of the times $\bf t$ and the corresponding polynomials (\ref{aba})
is reduced to the finite basis
\begin{equation}\label{dig}
 D_i(G)=\<G_i,G\>,\quad G_i=\la^i+g_1\la^{i-1}+\dots+g_i.
\end{equation}
Any given differentiation corresponding to an arbitrary $\mu$ can be
expanded over this basis. In particular, $\pd_x=D_0$, and the value $\mu=0$
corresponds to the differentiation $D_{n-1}$. In this latter case the
equations for the coefficients are of the form
\begin{equation}\label{gjtau}
 g_{1,t_{n-1}}=g_{n,x},\quad g_{2,t_{n-1}}=\<g_1,g_n\>,\quad\dots,\quad
 g_{n,t_{n-1}}=\<g_{n-1}, g_n\>.
\end{equation}
This system possesses rather many applications. For example, the case $n=2$
\[
 g_{1,t_1}=g_{2,x},\quad g_{2,t_1}=\<g_1,g_2\>
 \quad\Leftrightarrow\quad u_{t_1t_1}=\<u_x,u_{t_1}\>
\]
describes the Chaplygin model in gas dynamics.

In the general case, we choose the zeroes of $G(\la)$ as the dynamical
variables:
\[
 G(\la)=(\la-\g_1)(\la-\g_2)\cdots(\la-\g_n)
\]
then equation (\ref{sme}) at $\la=\g_i$ implies
\[
 D_\tau G(\la)\big|_{\la=\g_i}=\frac{G(\mu)G_x(\la)}{\mu-\la}\big|_{\la=\g_i}
 \quad\Leftrightarrow\quad
 D_\tau\g_i=\frac{G(\mu)}{\mu-\g_i}\g_{i,x},\quad i=1,\dots,n
\]
or
\begin{equation}\label{dtau}
  D_\tau\g_{1}=(\mu-\g_2)\cdots(\mu-\g_n)\g_{1,x},~\dots,~
  D_\tau\g_{n}=(\mu-\g_1)\cdots(\mu-\g_{n-1})\g_{n,x}.
\end{equation}
Comparing this with (\ref{dig}) at $\la=\g_i$ we obtain
\begin{equation}\label{mutau}
 D_\tau=\mu^{n-1}D_0+\mu^{n-2}D_1+\dots+D_{n-1}.
\end{equation}

The problem of solving (\ref{dtau}) admits various settings, and we will
construct the general solution following the paper \cite{EF}. In virtue of
the Theorem \ref{th:D12} and the formula (\ref{mutau}), the differentiations
$D_i$ mutually commute. We identify $g_1,\dots,g_n$ with the elementary
symmetric polynomials
\[
 g_1=-\sum\g_k,\quad g_2=\sum_{k<l}\g_k\g_l,\quad
 g_3=-\sum_{k<l<m}\g_k\g_l\g_m,~\dots
\]
and introduce the notations
\[
 g_{i,j}=g_i|_{\g_j=0}\quad\Rightarrow\quad
 g_{1,1}=-\g_2-\dots-\g_n,\quad \dots,\quad
 g_{n-1,n}=(-1)^n\g_1\cdots\g_{n-1}.
\]
This bring to the following statement.

\begin{statement}\label{th:digj}
Integration of the equations (\ref{dtau}) is equivalent to the integration
of  consistent system of $n(n-1)$ equations
\begin{equation}\label{digj}
 D_i(\g_j)=g_{i,j}D_0(\g_j),\quad i=1,\dots,n-1,\quad j=1,\dots,n.
\end{equation}
\end{statement}

It is worth to notice that the equation (\ref{digj}) with the number $(i,j)$
is obtained from the corresponding equation (\ref{dig}) under the
substitution $\la=\g_j$ and that
\[
 g_{i,j}=G_i(\g_j)=\g^i_j+\g^{i-1}_jg_1+\dots+g_i
        =-\g^{i-n}_j(\g^{n-i-1}_jg_{i+1}+\dots+g_n),\quad i<n.
\]
Notice that the zeroes $\g_j$ are Riemannian invariants, that is, this
change of variables brings each of $n-1$ quasilinear system of first order
equations (\ref{dig}) to the diagonal form (recall that (\ref{gjtau})
corresponds to $i=n-1$).

\begin{theorem}\label{th:polsol}
Let $D_i=\pd_{t_i}$, $i=0,\dots,n-1$. Then the general solution of the
overdetermined system (\ref{digj}) is given by the formula
\begin{equation}\label{dtg}
 \left[\begin{array}{c}
   dt_0\\ dt_1\\ \dots\\ dt_{n-1}
 \end{array}\right]=
 \left[\begin{array}{cccc}
  \g_1^{n-1}& \g_2^{n-1}& \dots &\g_n^{n-1} \\
  \g_1^{n-2}& \g_2^{n-2}& \dots &\g_n^{n-2} \\
                     &&\dots\dots &\\
           1&          1& \dots &   1
 \end{array}\right]
 \left[\begin{array}{c}
  \ka_1^{-1}(\g_1)d\g_1\\
  \ka_2^{-1}(\g_2)d\g_2\\
  \dots\\
  \ka_n^{-1}(\g_n)d\g_n
 \end{array}\right].
\end{equation}
\end{theorem}
\begin{proof}
Assume, in addition to (\ref{digj}), that the dynamics on $x$ of the zeroes
is defined by the system
\begin{equation}\label{d0}
 \g_{1,x}=R_1(\g_1,\dots,\g_n),~\dots,~ \g_{n,x}=R_n(\g_1,\dots,\g_n)
\end{equation}
and let us find the form of the functions $R_i$ by use of the conditions of
consistency with (\ref{digj}) which are obtained by cross-differentiation.
For example, the use of (\ref{dtau}) gives
\begin{gather*}
 D_0\g_1=R_1,\quad
 D_\tau\g_1=(\mu-\g_2)\cdots(\mu-\g_n)R_1 \quad\quad\Rightarrow\quad\\
 [D_\tau,D_0](\g_1)=D_\tau R_1-D_0(\mu-\g_2)\cdots(\mu-\g_n)R_1.
\end{gather*}
Substituting $\mu=\g_2,\dots,\g_n$ into the consistency condition
$[D_\tau,D_0](\g_1)=0$, we obtain
\[
 \frac{\pd\log R_1}{\pd \g_j}=\frac{1}{\g_1-\g_j},\quad
 \frac{\pd}{\pd\g_j}\log\Bigl(R_1\prod_{j=2}^n(\g_1-\g_j)\Bigr)=0,\quad
 j=2,\dots,n.
\]
The other consistency conditions give respectively
\[
 \frac{\pd\log R_i}{\pd \g_j}=\frac{1}{\g_i-\g_j},\quad
 \forall~i\ne j \quad\Rightarrow\quad
  R_i=\frac{\ka_i(\g_i)}{\prod'(\g_i-\g_j)}
\]
with $n$ arbitrary functions $\ka_1(\g_1),\dots,\ka_n(\g_n)$.

Apparently, the found additional dynamical system
\begin{equation}\label{ef}
 \g_{1,x}=\frac{\ka_1(\g_1)}{(\g_1-\g_2)\cdots(\g_1-\g_n)},\quad
 \g_{2,x}=\frac{\ka_2(\g_2)}{(\g_2-\g_1)(\g_2-\g_3)\cdots(\g_2-\g_n)},~\dots
\end{equation}
explicitly defines, together with (\ref{digj}), the Jacobi matrix $n\times
n$
\begin{equation}\label{jac}
 J=\frac{\pd(\g_1,\dots,\g_n)}{\pd(t_0,\dots,t_{n-1})}
\end{equation}
of the partial derivatives of $n$ zeroes with respect to $n$ independent
variables ($t_0=x$). One can straightforwardly check that $J$ coincide with
the inverse of the Vandermonde matrix $W$ from (\ref{dtg}), up to the left
multiplication by the diagonal matrix:
\[
 J=\diag(\ka_1(\g_1),\ka_2(\g_2),\dots,\ka_n(\g_n))W^{-1}.
\]
In order to finish the proof, we note that the formula (\ref{dtg}) rewritten
as
\begin{equation}\label{dtx}
 dt_{n-1}=\frac{d\g_1}{\ka_1(\g_1)}+\dots+\frac{d\g_n}{\ka_n(\g_n)},
 \quad\dots,\quad
 dx=\frac{\g_1^{n-1}d\g_1}{\ka_1(\g_1)}+\dots+\frac{\g_n^{n-1}d\g_n}{\ka_n(\g_n)}
\end{equation}
allows to solve in quadratures the problem of finding $t_0,\dots,t_{n-1}$ as
functions on $\g_1,\dots,\g_n$.
\end{proof}

Obviously, the degree of the polynomial $G_x(\la)$ is $(n-1)$ and equation
(\ref{ef}) and relations
\[
 G_x(\la)=-\g_{1,x}(\la-\g_2)\cdots(\la-\g_n)
  -\ldots-\g_{n,x}(\la-\g_1)\cdots(\la-\g_{n-1})
\]
imply that in order to calculate the Jacobian (\ref{jac}) and to integrate
the differential equations (\ref{dig}) it is sufficient to choose,
arbitrarily, the functions
\begin{equation}\label{kapj}
 G_x(\la)\big|_{\la=\g_j}=-\ka_j(\g_j),\quad j=1,\dots,n.
\end{equation}

\begin{corollary}
Theorem \ref{th:polsol} allows to find all derivatives $D_i(G)$ directly in
terms of $\g_1,\dots,\g_n$, without integrating differential equations
(\ref{dig}). In particular, $D_0(G)=G_x$ is defined by Lagrange
interpolation formula:
\[
 G_x(\la)=-\ka_1(\g_1)\frac{(\la-\g_2)\cdots(\la-\g_n)}
                              {(\g_1-\g_2)\cdots(\g_1-\g_n)}
    -\dots-\ka_n(\g_n)\frac{(\la-\g_1)\cdots(\la-\g_{n-1})}
                              {(\g_n-\g_1)\cdots(\g_n-\g_{n-1})}.
\]
\end{corollary}

The independent variable $t_{n-1}$ corresponding to $\mu=0$ can be replaced
with $\tau$. The transition from the variables $t_0=x$, $t_1,\dots,t_{n-1}$
to the variables $t_0$, $t_1,\dots,t_{n-2},\tau$ is defined by the
Jacobian
\[
 \frac{\pd(t_0,\dots,\tau)}{\pd(\g_1,\dots,\g_n)}=
 \frac{\pd(t_0,\dots,\tau)}{\pd(t_0,\dots,t_{n-1})}
 \frac{\pd(t_0,\dots, t_{n-1})}{\pd(\g_1,\dots,\g_n)}
\]
Since $d\tau=\mu^{n-1}dt_0+\mu^{n-2}dt_1+\dots+dt_{n-1}$ we obtain
\[
 d\tau=\sum\frac{\ka(\g_i)d\g_i}{\ka_i(\g_i)},\quad
   \ka(\g)=(\mu\g)^{n-1}+(\mu\g)^{n-2}+\dots+1.
\]
In order to apply Theorem \ref{th:polsol} to equation (\ref{mu12}) it is
sufficient to express the potential through the seroes $\g_j(x,\a,\b)$,
$j=1,2\dots$:
\[
 du=\frac{\g_1^nd\g_1}{\ka_1(\g_1)}+\dots
    +\frac{\g_n^nd\g_n}{\ka_n(\g_n)} \quad\Rightarrow\quad
 D_ju=g_j.
\]

%-------------------------------------------------------------------------------
\section{The map $G\to U$}\label{s:trace}

Let us consider the classification of the polynomial solutions $G(\la)$ of
equation (\ref{sme}) based on the map
\begin{equation}\label{ukg}
 G\to U:~ 4U=\frac{K(\la)}{G^2}-\frac{G_x^2}{G^2}+2\frac{G_{xx}}{G}.
\end{equation}
Notice, that the kernel of the map $G\to U$ is not trivial: it consists of
the quadratic polynomials in $x$ (\ref{abc}) (see Example \ref{ex:x^2}.)
Indeed, the equation (\ref{ukg}) after multiplying by $G^2$ and
differentiating yields the third order linear equation
\begin{equation}\label{***}
 2UG_x+U_xG=\frac12G_{xxx}.
\end{equation}
Therefore, if $U=0$ then $G_{xxx}=0$. Also, notice that in virtue of the
homogeneity of the equation (\ref{ukg})), the change
\[
 \ti G=\ka(\la)G,\quad \ti K=\ka^2K \quad\Rightarrow\quad \ti U=U
\]
does not change $U$. This allows to consider as equivalent the polynomial
solutions $\ti G=\ka(\la)G$ and $G$ which differ by a factor $\ka(\la)$
depending on $\la$ only.

The use of the map (\ref{ukg}) for the classification of the polynomial
solutions $G(\la)$ is related, above all, with the theory, developed by
S.P.~Novikov school, of the finite-gap potentials $U$ for the second order
spectral problems
\begin{equation}\label{psixx}
 \psi_{xx}=U(x,\la)\psi.
\end{equation}
In connection with these applications to the spectral theory the following
formulae are useful:
\begin{equation}\label{ugf}
 f^\pm\defeq \frac{G_x}{2G}\pm\frac{\sqrt{K(\la)}}{2G}
 \quad\Rightarrow\quad
 f^\pm_x+(f^\pm)^2=\frac{K(\la)}{4G^2}+\frac{G_{xx}}{2G}
   -\frac{G_x^2}{4G^2}=U,
\end{equation}
relating the equation (\ref{ukg}) with Riccati equation $f_x+f^2=U$. In
Appendix we discuss the {\em direct problem} of construction, for the given
potential $U$, of two special solutions $f=f^\pm$ of Riccati equation such
that
\[
 f^+-f^-=\frac{\sqrt{K(\la)}}{G}.
\]
This allows, in particular, to interpret the function $K(\la)$ in equation
(\ref{ukg}) as the Wronskian of two solutions $\psi^\pm$ of the spectral
problem (\ref{psixx}) related to $f=f^\pm$ by the standard rule
$(\log\psi^\pm)_x=f^\pm$.

The zeroes $\la=\g_i$ of the polynomial solutions $G(\la)$ of the equation
(\ref{sme}) lead, in general, to the poles of the potential defined by the
formula (\ref{ukg}). This restricts its use in the applications where the
analytic structure of $U(\la)$ is fixed a priori. The theorem below explains
how to get rid of these unwanted poles by the special choice of the
functions $\ka_j(\la)$ in equation (\ref{kapj}).

\begin{theorem}
The potential $U(\la,x)$ as function on $\la$ does not possess the moving
poles $\la=\g_j(x)$, $\g_{j,x}\ne0$ if and only if the following {\em
regularity conditions} are fulfilled:
\begin{equation}\label{kak}
 \ka_i(\la)=\pm\sqrt{K(\la)},\quad i=1,\dots,n.
\end{equation}
\end{theorem}
\begin{proof}
Let us deduce the equations analogous to (\ref{ef}) directly from the basic
equation (\ref{ukg}), using only the regularity condition and without use of
the Statement \ref{th:digj}. To do this substitute
\[
 \frac{1}{G}=\frac{\eps_1}{\g_1-\la}+\dots+\frac{\eps_n}{\g_n-\la},\quad
 \Gamma=\frac{G_x}{G}=\frac{\g_{1,x}}{\g_1-\la}+\dots+\frac{\g_{n,x}}{\g_n-\la},\quad
 \eps_i\defeq\prod_{j\ne i}\frac1{\g_i-\g_j}
\]
into the equation
\begin{equation}\label{ugg}
 4U=\frac{K(\la)}{G^2}+2\Gamma_x+\Gamma^2.
\end{equation}
Vanishing of the coefficients at the poles
\[
  (\la-\g_1)^{-2},\quad (\la-\g_2)^{-2},\dots\quad (\la-\g_n)^{-2}
\]
yields $n$ {\em Dubrovin equations} for $\g_1,\dots,\g_n$:
\begin{equation}\label{dbr}
 \g_{1,x}=\frac{\sqrt{K(\g_1)}}{(\g_1-\g_2)\cdots(\g_1-\g_n)},\quad\dots\quad
 \g_{n,x}=\frac{\sqrt{K(\g_n)}}{(\g_n-\g_1)\cdots(\g_n-\g_{n-1})}.
\end{equation}
Comparing these equations with (\ref{ef}) we see that the conditions
(\ref{kak}) are fulfilled.

It turns out that these conditions are not only necessary but also
sufficient for the regularity of the potential $U$ defined accordingly to
(\ref{ukg}) by solution $G$ of equations (\ref{ef}). Indeed, differentiation
of equations (\ref{dbr}) with respect to $x$ yields
\begin{equation}\label{gxx}
 \g_{i,xx}=\g^2_{i,x}\biggl(\frac12\frac{K'}{K}(\g_i)
   -\sum_{j\ne i}\frac{1}{\g_i-\g_j}\biggr)
   +\sum_{j\ne i}\frac{\g_{i,x}\g_{j,x}}{\g_i-\g_j},\quad i=1,\dots,n.
\end{equation}
On the other hand
\begin{equation}\label{k012}
 \frac{K(\la)}{G^2(\la)}=K_2+K_1+K_0,\quad
 K_2=\sum_1^n\frac{\eps_i^2 K(\g_i)}{(\la-\g_i)^2},\quad
 K_1=\sum_1^n\frac{\eps_i^2 K(\g_i) K_{1,i}}{\la-\g_i}
\end{equation}
where
\[
 K_{1,i}=\frac{d}{d\la}\log K(\la)\big|_{\la=\g_i}
  -2\sum_{j\ne i}\frac{1}{\g_i-\g_j},
\]
and $K_0$ is a regular function. Since, accordingly to (\ref{ugg}),
\[
 4U=\frac{K(\la)}{G^2}-\sum_1^n\frac{\g^2_{i,x}}{(\g_i-\la)^2}
  +2\sum_{i<j}\frac{\g_{i,x}\g_{j,x}}{(\g_i-\la)(\g_j-\la)}
  +2\sum_1^n\frac{\g_{i,xx}}{(\g_i-\la)},
\]
hence we prove that the first and second order poles are cancelled in virtue
of Dubrovin equations (\ref{dbr}).
\end{proof}

It follows from above that the solution $G$ is defined by the functions
$\ka_j(\la)$ in (\ref{dtx}), and if these functions are subjected to the
regularity condition (\ref{kak}) with an analytic function $K(\la)$ then the
potential $U$ is analytic in $\la$. The Liouville theorem says that a
function on $\la$ analytic in the whole extended complex plane is constant.
In particular (see (\ref{ukg})) the polynomial functions $K(\la)$ correspond
to the polynomial in $\la$ potentials $U$. This simple structure of the
potential $U$ is defined in this case by expansion in the neighborhood of
$\la=\infty$ of the function
\begin{equation}\label{hla}
 H(\la)=\frac{\la^n}{G}=1+\frac{h_1}{\la}+\frac{h_2}{\la^2}+\dots,\quad
 h_1=-g_1,\quad h_2=g_1^2-g_2,\dots
\end{equation}
and by the coefficients $c_i$ of the expansion
\[
 \la^{-m}K(\la)=4\Bigl(1+\frac{c_1}{\la}+\frac{c_2}{\la^2}+\dots\Bigr)
\]
where $m$ denotes the order of the pole $K(\la)$ at $\la=\infty$. Notice,
that the coefficients of the expansion (\ref{hla}) define the conservation
laws of the equations (\ref{dig}), since
\begin{equation}\label{dib}
 D_i(H)=D_i\Bigl(\frac{\la^n}{G}\Bigr)=\la^n\frac{GG_{i,x}-G_xG_i}{G^2}
  =HG_{i,x}+H_xG_i=D_0(H\, G_i).
\end{equation}
Thus, in the polynomial case the form of the potential $U$ is defined by the
first term in the formula (\ref{ukg}) and if the degree $m$ of the
polynomial $K(\la)$ exceeds the degree $2n$ of the polynomial $G^2(\la)$ by
1 or 2 then, correspondingly,
\[
 \frac{K(\la)}{4G^2}=
 \begin{cases}
  \la+2h_1+c_1+\dots, \\
  \la^2+\la(2h_1+c_1)+h_1^2+2h_2+2c_1h_1+c_2+\dots\,.
 \end{cases}
\]
Since the rest terms in the formula (\ref{ukg}) vanish at $\la=\infty$,
hence we obtain in the latter case, at $m=2n+2$,
\[
 U=\la^2+u_1\la+u_2,\quad u_1=2h_1+c_1,\quad u_2=h_1^2+2h_2+2c_1h_1+c_2.
\]
Analogously, in the case $m=2n+1$
\begin{equation}\label{u1}
 U=\la+u(x),\quad u(x)=c_1+2h_1=2\sum_{j=1}^n\g_j(x)-\sum_{i=1}^{2n+1}e_i
\end{equation}
where $e_i$, $i=1,2,\dots$ denote $2n+1$ zeroes of the polynomial $K(\la)$.
It is well known \cite{NMPZ} that in the case of the real potential the
necessary and sufficient condition of its regularity is that the initial
values of $\g_i$ lie in the restricted gaps of the spectrum:
\[
 e_1<e_2<g_1<e_3<\dots<e_{2n}<\g_n<e_{2n+1}.
\]

If the function $K(\la)$ possesses, for example, the first order pole at
$\la=0$ in addition to the pole of order $m=2n+1$ at infinity then the
equation (\ref{ukg}) defines the potential in the form
\begin{equation}\label{ito}
 U=\la+u+v/\la=\la+2h_1+c_1+\frac{1}{\la}(h_1^2+2h_2-\frac12h_{1,xx}).
\end{equation}
In this case the scheme of construction of particular solutions similar to
the presented above can be found in the paper \cite{2006}, see also
\cite{ush}.

The observation that the map (\ref{ukg}) is factorizable gives additional
possibilities of the classification of the polynomial solutions $G(\la)$.
Consider the intermediate map $G\to V$ defined as follows
\begin{equation}\label{vkg}
 V=\frac{G_x+\ka(\la)}{2G}.
\end{equation}
Then, in virtue of (\ref{ugf}),
\[
 U=V_x+V^2,\quad K(\la)=\ka^2(\la).
\]
Moreover, the regularity condition $V$ is equivalent, due to (\ref{kapj}),
to the relations $\ka_i=-\ka$, $i=1,\dots,n$ (cf (\ref{kak})). Like in the
case (\ref{ukg}), the regularity condition guarantees that if $\ka(\la)$ is
polynomial then $V$ is polynomial as well. In particular, the choice of the
polynomial $\ka(\la)$ of $m=n+1$ degree we find, analogous to (\ref{u1})
that
\begin{equation}\label{v1}
 V=\la+v(x),\quad v=-g_1,\quad U=\la^2+2v\la+v^2+v_x.
\end{equation}

In order to illustrate the variety of the applications of the regularity
conditions of the form (\ref{kak}) and Theorem \ref{th:polsol} we present a
list of scalar partial differential equations which correspond to equations
(\ref{dig}) in the cases (\ref{u1}) and (\ref{v1}). Since the analytic
structure of $U$ and $V$ in our case is essentially simpler than the one of
$G$, hence it is natural to rewrite the equations (\ref{dig}) in terms of
$U$ and $V$. To do this, it is sufficient to notice that
\begin{equation}\label{dfv}
 G_t=RG_x-GR_x\quad\Rightarrow\quad V_t=\Bigl(RV-\frac12R_x\Bigr)_x
\end{equation}
and analogously
\begin{equation}\label{dfu}
 G_t=RG_x-GR_x\quad\Leftrightarrow\quad f_t=\Bigl(R f-\frac12R_x\Bigr)_x
 \quad\Rightarrow\quad
 U_t=2UR_x+U_xR-\frac12R_{xxx}.
\end{equation}

\begin{example}\label{ex:kdv}
The choice $R=G_1=\la+g_1$ leads, in the case (\ref{v1}), (\ref{dfv}) to the
Burgers equation
\begin{equation}\label{brgs}
 v_t=v_{xx}+2vv_x=[(\la+g_1)(\la+v)-g_{1,x}]_x,
\end{equation}
and in the case (\ref{u1}), (\ref{dfu}) to the Korteweg-de Vries equation
\begin{equation}\label{kdv}
 4u_t=u_{xxx}-6uu_x+\eps u_x.
\end{equation}
On the other hand, the use in (\ref{dfu}) of the equation for $f$, allows
easily to obtain the modified KdV equation. Indeed, the second equation
(\ref{dfu}) gives at $R=G_1=\la+g_1$
\[
 D_1(f)=D_x\Bigl(\la f+fg_1-\frac12g_{1,x}\Bigr),\quad f_x+f^2=U=\la-2g_1+c_1.
\]
We find, by substitution of the expression of $g_1$ through $f$ into the
first equation under consideration:
\begin{equation}\label{mkdv}
 4f_t=D_x(f_{xx}-2f^3+6\la f)+\eps f_x.
\end{equation}
The relation between the solutions of equations (\ref{kdv}), (\ref{mkdv})
defines the  Miura map $u=f_x+f^2+\la$ depending on $\la$.

Finally, we demonstrate that the original system of equations (\ref{gjtau})
leads, in the case $U=\la+u$, directly to sinh-Gordon equation:
\begin{equation}\label{sinh}
 \ph_{x\tau}=\frac12k_1e^{-\ph}-2e^\ph,\quad e^\ph=G(0),\quad
 k_1=\frac{dK}{d\la}\big|_{\la=0},\quad (K(0)=0).
\end{equation}
Indeed, we obtain from (\ref{gjtau}), (\ref{ukg}), correspondingly,
\begin{equation}\label{tig}
 g_\tau=\<\ti g,g\>=\ti gg_x-g\ti g_x,\quad 4(\la+u)G^2=K(\la)-G_x^2+2G_{xx}G
\end{equation}
where
\[
 g=G(0),\quad \ti g=\frac{dG}{d\la}\Big|_{\la=0}.
\]
We find by differentiating with respect to $\la$ of the second equation
(\ref{tig}) and then eliminating $u$:
 \[
  2g^2-\frac12k_1+k_0\frac{\ti g}{g}+\<\ti g,g\>_x
  -(\log g)_x\<\ti g,g\>=0,\quad k_0\defeq K(0).
\]
Next, setting $k_0=0$ and replacing $\<\ti g,g\>$ with $g_\tau$ bring the
latter equation to (\ref{sinh}), after the obvious transformations.
\end{example}

To conclude the section we consider briefly the applications of equation
(\ref{sme}) to the question of integrability of the analogs of equations
(\ref{dfu}), (\ref{dfv}) discussed in the Example \ref{ex:kdv} obtained by
the change of polynomials by solutions $G(\la)$ of the form (\ref{ab}).
Notice that if we increase the degree of the polynomials $G(\la)$ then the
solutions corresponding to the degree $n$ remain automatically the solutions
of the system (\ref{dig}) after the change $n$ by $n+1$. Therefore it is
natural to suggest that many formulae of the previous Section, including
(\ref{ugf})--(\ref{dfu}), remain valid in the general case (\ref{ab}). The
proof follows from the comparison of the equations (\ref{dig}) and
(\ref{dna}) under the change $G(\la)\to A(\la)$ where
\begin{equation}\label{aser}
 A=1+\sum_{j=1}^{\infty}\la^{-j}a_j
  =1+\frac{a_1}{\la}+\frac{a_2}{\la^2}+\dots\,.
\end{equation}
The main difference between the equations for the coefficients of the series
(\ref{aser}) under consideration
\begin{equation}\label{dnaa}
 D_n A(\la)=\<A_n(\la),A(\la)\>,\quad n=1,2,3,\dots
\end{equation}
and the corresponding equations (\ref{dig}) is the problem of closing of the
infinite chain of equations of the form (\ref{dta}).\footnote{From the
algebraical point of view this problem is, probably, equivalent to the
choice of the corresponding subalgebras of the algebra (\ref{da}).} The
papers \cite{alonso_shabat} are devoted to the problems which arise after
the change of regularity conditions by the conditions of finiteness of the
series
\begin{equation}\label{UV}
 U=\la^m+\sum_{j\ge1}\la^{m-j}u_j,\quad
 V=\la^{m'}+\sum_{j\ge1}\la^{m'-j}v_j
\end{equation}
obtained from (\ref{aser}) via the maps (\ref{ukg}) and (\ref{vkg}). These
papers contain also some interesting examples of the nonstandard closing of
equations (\ref{dnaa}).

\appendix
%-------------------------------------------------------------------------------
\section{Appendix}

%-------------------------------------------------------------------------------
\subsection{Riccati equation}

We prove the following elementary statement in order to establish the
relation of the map (\ref{ukg}) and equation (\ref{ugf}) with the Schwarz
derivative and Riccati equation.

\begin{statement}
Let $\psi_1$, $\psi_2$ be two linearly independent solutions of the second
order equation $\psi_{xx}=U\psi$, then the functions
\[
 A_1=\psi_1^2,\quad A_2=\psi_2^2,\quad A_3=\psi_1\psi_2
\]
form the basis of the solution space of the third order equation (cf
(\ref{***})):
\[
 A_{xxx}=4UA_{x}+2U_x A.
\]
Moreover the function $\ph=\psi_1/\psi_2$ satisfies Schwarz equation
\begin{equation}\label{schwarz}
 \frac{3\ph^2_{xx}}{4\ph^2_x}-\frac{\ph_{xxx}}{2\ph_x}=U(x),
\end{equation}
and the function $A=\psi_1\psi_2$ satisfies equation
\begin{equation}\label{**}
 4U(x)A^2+A_x^2-2AA_{xx}= w^2
\end{equation}
where $w$ is the Wronskian
$w=\<\psi_1,\psi_2\>=\psi_1\psi_{2,x}-\psi_{1,x}\psi_2$.
\end{statement}
\begin{proof}
The Wronskian of $A_1$, $A_2$, $A_3$ is
\[
 W=\<A_1,A_2,A_3\>=(\psi_1\psi_{2,x}-\psi_2\psi_{1,x})^3
  =\<\psi_1,\psi_2\>^3.
\]
therefore, $W=\const\ne0$ and the functions $A_i$ are linearly independent.
Let us denote
\[
 \ph=\frac{\psi_1}{\psi_2},\quad f_j=\frac{\psi_{j,x}}{\psi_j}.
\]
It is not difficult to prove that
\[
 \ph_x=\frac{\<\psi_2,\psi_1\>}{\psi_2^2}
      =\frac{w}{\psi_2^2}, \quad
 \frac{\ph_{xx}}{\ph_x}= -2\frac{\psi_{2x}}{\psi_2}=-2f_2
\]
and therefore equation (\ref{schwarz}) follows from the Riccati equation
$f_{2,x}+f_2^2=U$.

Next, it is easy to check that
\[
 \frac{w}{A_3}=f_2-f_1,\quad \frac{A_{3,x}}{A_3}=f_2+f_1,
\]
and therefore
\[
 f_1=\frac{A_{3,x}-w}{2A_3},\quad f_2=\frac{A_{3,x}+w}{2A_3}.
\]
The substitution of these expressions for $f_j$ into Riccati equation brings
to equation (\ref{**}) with $A=A_3$ and then the differentiation with
respect to $x$ yields (\ref{***}). The proof of the fact that $A_1$, $A_2$
also satisfy (\ref{***}) is analogous, in these cases $w=0$ in (\ref{**}).
\end{proof}

Let us rewrite equation (\ref{**}) in the form
\[
  4U+\frac{A_x^2}{A^2}-\frac{2 A_{xx}}{A}=\frac{w^2}{A^2}
\]
and denote $H=1/A$. Then we find that
\begin{equation} \label{modSchwarz}
  U=\frac{3H_x^2}{4H^2}-\frac{H_{xx}}{2H}+w^2H^2.
\end{equation}
If follows immediately from the comparison of this equation with Schwarz
equation (\ref{schwarz}) that it correspond to the special case
(\ref{modSchwarz}) with $w=0$ and $H=\ph_x$.

Therefore we prove that the change
\[
  F\defeq f_1-f_2=\frac{w}{G},\quad
  \frac{3F_x^2}{4F^2}-\frac{F_{xx}}{2F}+\frac{F^2}{4}=U
\]
brings to equation (\ref{ukg}) used for the definition of the map $G\to U$
(in Section \ref{s:trace}). Taking into account the obvious relation between
the formula (\ref{modSchwarz}) with the Schwarz derivative one may say that
this transition from $G$ to $U$ is equivalent in some sense to the
computation of the Schwarz derivative of $F=G^{-1}$.

Now consider the {\em direct problem} of construction of $G$ on the given a
priori potential $U$. We assume that the potential is defined as the formal
series (\ref{UV}):
\[
 U=\la^m+\sum_{j\ge1}\la^{m-j}u_j,\quad m>0.
\]
It turns out that the difficulties which appear if we try to solve the
Riccati equation explicitly disappear on the level of the formal series
(\ref{aser}). For the sake of simplicity we restrict ourselves by the
particular case $m=2$.

\begin{lemma}
The Riccati equation with the potential
\begin{equation}\label{u2}
 U(x,k)=\la^2+\la u_1(x)+u_2(x)+\la^{-1}u_3+\dots
\end{equation}
has exactly two solutions $f=f^\pm$ in the form of the formal power series
in $\la$:
\begin{equation}\label{fpm}
 f^+= \la+\frac12u_1+f^+_1\la^{-1}+\dots,\quad
 f^-=-\la-\frac12u_1+f^-_1\la^{-1}+\dots\ .
\end{equation}
\end{lemma}
\begin{proof}
The formulae for the coefficients of $f\pm_0$ are obtained by substitution
of the series $f=\sum^\infty_{j=j_0}\la^{-j}f_j$ into the equation
$f_x+f^2=U(x,\la)$. Collecting the coefficients at the powers of $\la$ we
find $(f^{\pm}_{-1})^2=1$, $j_0=-1$. The next coefficients of $f^+_j$ and
$f^-_j$ are calculated recursively along the relations
\[
 2f^+_{j+1}+f^+_{j,x}+\sum_{j'+j''=j}f^+_{j'}f^+_{j''}=u_j,\quad
 2f^-_{j+1}=f^-_{j,x}+\sum_{j'+j''=j} f^-_{j'} f^-_{j''}-u_j\qquad (j>1).
\]
Therefore the coefficients of the series $f^{\pm}_j$ are uniquely defined
differential polynomials on $u_j$.
\end{proof}

%-------------------------------------------------------------------------------
\subsection{The consistent triples of equations}\label{app:triples}

The equation (\ref{mu12}) gives an example of consistent triple of the
form
\begin{equation}\label{fgh}
\begin{aligned}
 u_{\xi\eta}  &=f(u,u_\xi,u_\eta, u_x,u_{x\xi},u_{x\eta}),\\
 u_{\xi\zeta} &=g(u,u_\xi,u_\zeta,u_x,u_{x\xi},u_{x\zeta}),\\
 u_{\eta\zeta}&=h(u,u_\eta,u_\zeta,u_x,u_{x\eta},u_{x\zeta})
\end{aligned}
\end{equation}
(here and further on $(\tau_1,\tau_2,\tau_3)=(\xi,\eta,\zeta)$). A more
symmetric example is given by the triple
\begin{equation}\label{ex2}
 (\mu_j-\mu_i)u_xu_{\tau_i\tau_j}-\mu_ju_{\tau_i}u_{x\tau_j}
  +\mu_iu_{\tau_j}u_{x\tau_i}=0,
\end{equation}
where the variable $x$ is actually on the equal footing with the other ones.
Notice that this equation is invariant under the changes $u=a(\tilde u)$,
$x=b(\tilde x)$, $\tau_i=c_i(\tilde\tau_i)$ with arbitrary functions, so
that a particular solution is $u=a(b(x)c(\tau_i)d(\tau_j))$.

Equations (\ref{mu12}) and (\ref{ex2}) are in close relation. Indeed,
(\ref{ex2}) is obtained from (\ref{mu12}) after eliminating the derivatives
$u_{x\tau_i}$ and the choice of one of $\tau_i$ as the new $x$. Moreover,
the equation (\ref{ex2}) can be rewritten in the form
\[
  \left(\frac{u_{\tau_i}}{\mu_iu_x}\right)_{\tau_j}
 =\left(\frac{u_{\tau_j}}{\mu_ju_x}\right)_{\tau_i}
\]
and the substitution $v_{\tau_i}=u_{\tau_i}/(\mu_iu_x)$ brings to the
equation (\ref{mu12}) again:
\[
 (\mu^{-1}_i-\mu^{-1}_j)v_{\tau_i\tau_j}=v_{\tau_i}v_{x\tau_j}-v_{\tau_j}v_{x\tau_i}.
\]

The classification problem of the consistent systems of the form (\ref{fgh})
seems not actual until the integration methods are not well understood.
Nevertheless, the simple preliminary analysis shows that the class of such
systems can be rich enough. Let us assume that the right hand sides are
quasilinear and do not depend on $u$ explicitly. This subclass of equations
contains the following family of consistent triples.

\begin{statement}
Let $(X,\ti X)$, $(A,\ti A)$, $(B,\ti B)$, $(C,\ti C)$ be solutions of the
ODE system
\begin{equation}\label{kk}
 X'=k_1X^2+k_2X\ti X+k_3\ti X^2,\quad \ti X'=k_4X^2+k_5X\ti X+k_6\ti X^2,
\end{equation}
and the following functions are not identically zero:
\begin{alignat*}{3}
 a(u_\eta,u_\zeta)&=B(u_\eta)\ti C(u_\zeta)-\ti B(u_\eta)C(u_\zeta),&\qquad
 p(u_\xi,u_x)&=A(u_\xi)\ti X(u_x)-\ti A(u_\xi)X(u_x),\\
 b(u_\zeta,u_\xi)&=C(u_\zeta)\ti A(u_\xi)-\ti C(u_\zeta)A(u_\xi),&
 q(u_\eta,u_x)&=B(u_\eta)\ti X(u_x)-\ti B(u_\eta)X(u_x),\\
 c(u_\xi,u_\eta)&=A(u_\xi)\ti B(u_\eta)-\ti A(u_\xi)B(u_\eta),&
 r(u_\zeta,u_x)&=C(u_\zeta)\ti X(u_x)-\ti C(u_\zeta)X(u_x).
\end{alignat*}
Then the following equations are consistent:
\[
\begin{aligned}
 c(u_\xi,u_\eta)u_{\xi\eta}&=p(u_\xi,u_x)u_{x\xi}-q(u_\eta,u_x)u_{x\eta},\\
 a(u_\eta,u_\zeta)u_{\eta\zeta}&=q(u_\eta,u_x)u_{x\eta}-r(u_\zeta,u_x)u_{x\zeta},\\
 b(u_\zeta,u_\xi)u_{\zeta\xi}&=r(u_\zeta,u_x)u_{x\zeta}-p(u_\xi,u_x)u_{x\xi}.
\end{aligned}
\]
\end{statement}

The linear changes bring the second equation of the system (\ref{kk}) to the
form
\[
 \ti X'=kX\ti X \qquad\text{or}\qquad \ti X'=k\ti X^2,
\]
and the general solution can be found in quadratures. In particular, the
system (\ref{ex2}) correspond to the functions $A_i=\mu_i/u_{\tau_i}$, $\ti
A_i=1/u_{\tau_i}$, $X=\d/u_x$, $\ti X=1/u_x$, which solve the system
$X'=-X\ti X$, $\ti X'=-\ti X^2$. In this case one of the integration
constants can be neglected without loss of generality while the second one
plays the role of parameter in the resulting system. In the examples
\begin{gather*}
  (u_{\tau_i}-u_{\tau_j})u_{\tau_i\tau_j}
 +(u_{\tau_j}-u_x)u_{x\tau_j}+(u_x-u_{\tau_i})u_{x\tau_i}=0,\\
 u_{\tau_i\tau_j}=\frac{u_{x\tau_i}-u_{x\tau_j}}{e^{u_{\tau_i}}-e^{u_{\tau_j}}}
\end{gather*}
both integration constants are not essential.

%-------------------------------------------------------------------------------
\paragraph{Acknowledgements.}
This work has been supported by RFBR grants \#04-01-00403,
06-01-92051-KE\_a.

%-------------------------------------------------------------------------------

\end{document}